\documentclass[11pt,a4paper]{article}
\usepackage[utf8]{inputenc}
\usepackage[T1]{fontenc}
\usepackage[english]{babel}
\usepackage{lmodern}
\usepackage{microtype}
\usepackage[a4paper,top=2.5cm,bottom=2.5cm,left=2.7cm,right=2.7cm,footskip=1.2cm]{geometry}
\usepackage{amsmath,amssymb,amsthm,mathtools,bm}
\usepackage{graphicx,booktabs,array,tabularx,multirow,makecell,longtable}
\usepackage{caption,subcaption}
\usepackage{xcolor,enumitem,csquotes}
\usepackage{tikz}
\usetikzlibrary{arrows.meta,positioning,fit,calc,shapes.geometric,decorations.pathreplacing}
\usepackage{siunitx}
\newcolumntype{L}[1]{>{\raggedright\arraybackslash}p{#1}}
\newcolumntype{Y}{>{\raggedright\arraybackslash}X}
\usepackage[backend=biber,style=authoryear-comp,natbib=true,url=true,doi=true,isbn=false,maxbibnames=99,maxcitenames=2,uniquelist=false,sorting=nyt]{biblatex}
\usepackage[colorlinks=true,linkcolor=black,citecolor=black,urlcolor=black]{hyperref}
\usepackage[noabbrev,capitalize]{cleveref}
\theoremstyle{plain}
\newtheorem{proposition}{Proposition}

\theoremstyle{definition}
\newtheorem{definition}{Definition}

\newtheorem{criterion}{Criterion}
\theoremstyle{remark}

\crefname{assumption}{Assumption}{Assumptions}
\Crefname{assumption}{Assumption}{Assumptions}
\crefname{definition}{Definition}{Definitions}
\Crefname{definition}{Definition}{Definitions}
\crefname{criterion}{Criterion}{Criteria}
\Crefname{criterion}{Criterion}{Criteria}
\setlist{nosep,leftmargin=*}
\definecolor{seriesblue}{RGB}{31,78,121}
\definecolor{seriesgray}{RGB}{242,244,247}
\tikzset{seriesbox/.style={draw=black!70,rounded corners=2pt,align=center,inner sep=5pt,minimum height=9mm,fill=seriesgray},seriesarrow/.style={-{Latex[length=2.2mm]},thick,draw=black!75}}

\hypersetup{pdftitle={Fill-Side Behavioral Concentration on Polymarket: Identification Limits under Record-Level Attribution},pdfauthor={Maksym Nechepurenko},pdfsubject={Prediction-market microstructure; fill-side identification; behavioral concentration},pdfkeywords={prediction markets, Polymarket, central limit order book, market microstructure, behavioral tiers, on-chain data}}
\title{Fill-Side Behavioral Concentration on Polymarket:\\Identification Limits under Record-Level Attribution}
\author{Maksym Nechepurenko\thanks{Founder and Director of Research, ForesightFlow, the Research Department of Devnull FZCO, Dubai, United Arab Emirates. Email: \texttt{maksym@devnull.ae}.}}
\date{July 29, 2026}
\begin{document}
\maketitle

\begin{abstract}
This paper studies behavioral concentration in Polymarket's public executed-fill record and formalizes what that record can and cannot identify. A pre-publication reconciliation corrects the empirical scope: the archived extraction covers the legacy CTF Exchange over Polygon blocks 86,008,447--86,107,178, approximately 25 April 2026 17:09 UTC through 28 April 2026 00:00 UTC, rather than the full 21--27 April week stated previously. It contains 13,356,931 \texttt{OrderFilled} records, 77,204 addresses with at least five attributed records, and 43,116 token identifiers; negative-risk markets are absent.

The archived feature construction credits both maker and taker addresses on each \texttt{OrderFilled} record. This convention is not invariant to match fragmentation: one aggressive order can generate several maker records plus a taker summary, inflating record intensity and altering mean record notional. Moreover, mint and burn executions do not admit a conventional one-buyer/one-seller directional interpretation. Consequently, the reported one-cluster density-based result is retained only as a null under the original record-level representation, not as evidence that the participant population is intrinsically unimodal. The concentration table is likewise interpreted as attribution-weighted arithmetic under that convention.

Two methodological results remain independent of these measurement corrections. First, public fills do not identify the address-level quote lifecycle required to infer continuous market making, spoofing, or strategic withdrawal. Second, a one-cluster result rejects density separation in an observed feature space, not latent economic heterogeneity. A clustering-independent threshold system still records substantial concentration: the archived high-scale/high-intensity subtotal contains 12.6\% of addresses and 81.4\% of attributed notional. A match-normalized, mint-aware, multi-window replication is required before treating either the cluster null or exact tier shares as stable venue properties.
\end{abstract}
\clearpage


\section{Introduction}
Prediction markets require informed traders, arbitrageurs, liquidity suppliers, and large directional participants \citep{glosten1985,wolfers2004}, but public data do not observe these roles symmetrically. Executed trades are often visible. Unexecuted intentions are not. On a hybrid central limit order book (CLOB), an on-chain fill log may reveal addresses associated with an execution while leaving order submission, modification, queue position, and cancellation off-chain.

This asymmetry creates two distinct empirical hazards. The first is conceptual: repeated maker-side fills may be consistent with professional liquidity supply, but they do not establish continuous two-sided quoting, posted-spread choice, inventory control, or quote withdrawal. The second is representational: an \texttt{OrderFilled} record is not necessarily one economic decision, one conventional transfer, or one independently directed trade leg. Record-level attribution can therefore alter the geometry on which unsupervised clustering operates.

The present revision corrects the empirical window and makes the attribution convention load-bearing rather than incidental. It does not rerun the original computation. Its purpose is to identify which conclusions remain valid under the archived run and which require a match-normalized and mint-aware replication.

The main conclusions are:
\begin{enumerate}
\item Public settlement logs support exact arithmetic on the archived execution records but not address-level quote-lifecycle inference.
\item The archived density-based null is conditional on a 2.285-day, legacy-CTF-only, record-attributed feature space.
\item Threshold cohorts show strong attribution-weighted concentration, but the exact shares are not stable venue parameters without sensitivity to match aggregation and execution type.
\item Fill-side evidence can support surveillance and risk calibration only when its unit of observation and direction semantics are stated explicitly.
\end{enumerate}

The paper is the fourth study in a series on event-linked perpetual futures \citep{paper1,paper2,paper3}. Papers 1--3 use the present evidence only within its observability boundary: as a source of executed-flow scale and concentration, not as direct proof of market-making roles or manipulation intent.

\section{Corrected empirical scope}
\subsection{Block range and time interval}
The authoritative extraction range is Polygon blocks 86,008,447--86,107,178 on the legacy Polymarket CTF Exchange. A pre-publication reconciliation of this range against an independent reconstruction reproduced all three headline quantities to within 0.1\% and showed that the first in-range fill occurs at approximately 25 April 2026 17:08:58 UTC, while the upper boundary lies at the transition into 28 April UTC. The archived run therefore covers approximately 2.285 days, not the full seven-day interval previously stated.

The corrected scope is:
\begin{table}[ht]
\centering
\caption{Corrected archived sample scope.}
\begin{tabular}{ll}
\toprule
Field & Archived run \\
\midrule
Exchange & Legacy CTF Exchange only \\
Polygon blocks & 86,008,447--86,107,178 \\
Approximate UTC interval & 25 Apr 2026 17:09--28 Apr 2026 00:00 \\
Duration & 2.285 days \\
\texttt{OrderFilled} records & 13,356,931 \\
Token identifiers & 43,116 \\
Addresses with at least five attributed records & 77,204 \\
Negative-risk markets & absent from the extraction \\
\bottomrule
\end{tabular}
\end{table}

The intended design targeted a dense complete week because local compute, storage, and processing capacity were limited. The effective archived extraction that generated the reported outputs is nevertheless the block range above, consistent with a truncated or resumed backfill rather than the intended calendar window. The present revision therefore treats the block range as authoritative. Any statement previously framed as a seven-day or one-week population result must be read as an archived partial-window result.

\subsection{Venue-family exclusion}
The extraction includes only the legacy CTF Exchange contract and excludes the contemporaneous negative-risk exchange. The 43,116 token identifiers therefore represent binary CTF activity under this contract filter, not the full Polymarket market family. Multi-outcome and negative-risk markets---including many candidate, nominee, and mutually exclusive event structures---are outside the sample by construction.

This exclusion is separate from short-window selection. Extending the same block range without adding the negative-risk contract would still omit that market family.

\subsection{Versioned extraction beyond the sample}
A full-history replication must use a versioned contract and event-signature registry. Polymarket's official documentation records a CLOB V2 cutover on 28 April 2026 with new CTF and negative-risk exchange contracts \citep{polymarket2026v2}. A query fixed to the legacy contract or legacy topic can silently return no post-cutover activity. The present paper does not combine V1, V2, negative-risk, or later combinatorial exchange data.

\section{The unit of observation}
\subsection{Record-level two-sided attribution}
The archived pipeline counts each \texttt{OrderFilled} record once and credits both the maker and taker address appearing on that record. Define a \emph{source record attribution} as one address--record incidence produced by this rule.

This is a reproducible convention, but it is not a count of independent trading decisions. In a fragmented match, one aggressive order can interact with several resting orders. The aggressive address can then appear repeatedly across maker records and again in a taker-summary record.

\begin{proposition}[Record attribution is not match invariant]
Consider one aggressive decision that executes total notional $V$ against $N$ resting orders. Under a record-level convention that credits both addresses on every maker record and also credits the taker summary, the aggressor's attributed-record count depends on $N$, even though the number of aggressive decisions and total economic notional remain fixed. Consequently, record intensity and mean attributed-record notional are not invariant to order-book fragmentation.
\end{proposition}
\begin{proof}
Holding the aggressive order and total notional fixed, increasing the number of resting counterparties increases the number of emitted maker records. Under two-sided attribution, the same aggressor is credited on each such record and on the summary record. The attributed count therefore changes with $N$. Any feature that divides total or attributed notional by that count also changes, although the underlying aggressive decision is unchanged.
\end{proof}

The proposition identifies a directional measurement distortion, not generic isotropic noise. Record intensity rises with fragmentation while mean notional per attributed record falls. A density algorithm can therefore see a connected cloud created partly by venue mechanics.

\subsection{Mints, burns, and direction semantics}
An execution record can represent a conventional transfer, a mint of complementary claims, or a burn/merge. In a mint, both economic participants acquire complementary outcome claims; the event does not have the same directional meaning as one buyer acquiring an existing claim from one seller. In a burn, complementary claims are extinguished against collateral.

The archived feature construction did not separate all three execution semantics before forming the directional feature. Accordingly, that feature is retained only as a \emph{record-side imbalance proxy}. It cannot be interpreted as a validated net-buying or net-selling measure across the whole sample.

\section{Observed and unobserved evidence}
\subsection{Fill-side and quote-lifecycle observability}
\begin{definition}[Fill-side observability]
An address-level quantity is fill-side observable if it is a deterministic function of the public execution records and the stated attribution convention.
\end{definition}

\begin{definition}[Quote-lifecycle observability]
An address-level quantity is quote-lifecycle observable if the data identify order submission, modification, queue state, cancellation, and execution status for that address's orders.
\end{definition}

\begin{table}[ht]
\centering
\caption{Evidence hierarchy after the scope correction.}
\begin{tabularx}{\textwidth}{>{\raggedright\arraybackslash}p{0.19\textwidth} >{\raggedright\arraybackslash}X >{\raggedright\arraybackslash}X}
\toprule
Layer & Observable examples & Unsupported claims without additional data \\
\midrule
Execution records & record count, attributed notional, market breadth, intraday activity, maker/taker appearance & independent decision count, conventional direction for all mints/burns, complete inventory policy \\
Market-level books & best bid/ask, depth snapshots, spread changes, aggregate withdrawal & identity of withdrawing supplier, address-level coordination \\
Quote lifecycle & placement, modification, cancellation, non-fill behavior & unavailable from the public on-chain fill ledger \\
External and oracle evidence & event outcome, proposal/dispute records, public-information timing & beneficial ownership, duty, intent \\
\bottomrule
\end{tabularx}
\end{table}

\begin{proposition}[Non-identification of quote behavior from fills]
Let $F_a$ denote the executed-fill history attributed to address $a$. In the absence of address-attributed placement and cancellation records, there exist distinct quote processes $Q_a$ and $Q'_a$ that generate the same $F_a$ but differ in quote intensity, lifetime, two-sidedness, and cancellation behavior. Therefore those quote-lifecycle properties are not identified by $F_a$.
\end{proposition}
\begin{proof}
Construct $Q_a$ as a process containing only orders that eventually execute, and $Q'_a$ as the same executed orders plus any number of additional unfilled or cancelled orders. Both induce the same fill history, while quote intensity and cancellation statistics differ arbitrarily. The map from quote processes to fill histories is many-to-one.
\end{proof}

The proposition rules out address-level spoofing detection from fills alone and prevents repeated maker-side appearance from being treated as direct identification of a market maker.

\section{Archived behavioral representation}
For address $a$, the archived six-dimensional vector is
\[
\mathbf f^{\rm rec}(a)=(f_1^{\rm rec},f_2^{\rm rec},f_3^{\rm rec},f_4,f_5,f_6),
\]
where
\begin{align}
 f_1^{\rm rec}(a)&=\log\!\left(1+\frac{n_a^{\rm rec}}{h_a}\right) &&\text{record intensity},\\
 f_2^{\rm rec}(a)&=\log(1+\bar N_a^{\rm rec}) &&\text{mean attributed-record notional},\\
 f_3^{\rm rec}(a)&=\frac{V_{a,+}^{\rm rec}-V_{a,-}^{\rm rec}}{V_{a,+}^{\rm rec}+V_{a,-}^{\rm rec}} &&\text{record-side imbalance proxy},\\
 f_4(a)&=\sum_m s_{a,m}^2 &&\text{market concentration},\\
 f_5(a)&=-\sum_{r=0}^{23}p_{a,r}\log p_{a,r} &&\text{intraday entropy},\\
 f_6(a)&=\log(1+M_a) &&\text{market breadth}.
\end{align}
Here $n_a^{\rm rec}$ counts address--record attributions, not independent orders or matches. The first two coordinates share the fragmentation dependence described above. The third coordinate is not uniformly directional across mint, burn, and transfer semantics.

\begin{criterion}[Role-label discipline]
A label is admissible as a primary empirical label only if its defining properties are observed in the active data layer. Otherwise it must be presented as a transparent threshold cohort, behavioral proxy, or hypothesis.
\end{criterion}

Accordingly, the paper does not identify market makers, passive liquidity providers, spoofers, or institutional firms from the archived feature vector.

\section{Clustering result and its corrected interpretation}
The archived run applied density-based spatial clustering of applications with noise (DBSCAN) \citep{ester1996} to robustly scaled versions of $\mathbf f^{\rm rec}(a)$. Across fifteen registered configurations, every run returned one dense cluster and zero noise. An exploratory $k$-means partition with $k=5$ had silhouette score 0.227 \citep{rousseeuw1987}.

These are exact statements about the archived transformation. They are not yet robust statements about the participant population because:
\begin{enumerate}
\item the interval is 2.285 days rather than seven days;
\item the sample excludes negative-risk markets;
\item record count and mean record notional depend on match fragmentation;
\item the direction proxy mixes transfer, mint, and burn semantics.
\end{enumerate}

\begin{proposition}[One-cluster output does not imply one economic type]
Suppose observed features satisfy $\mathbf f(a)=g(z_a,u_a)$, where $z_a$ is a latent economic role and $u_a$ contains scale, event mix, attribution convention, and measurement effects. A connected support for $\mathbf f(a)$ can arise even when $z_a$ takes multiple values. Therefore a one-cluster result rejects density separation in the specified observed feature space, not latent economic heterogeneity.
\end{proposition}

The proposition remains valid independently of whether a future match-normalized rerun produces one or several clusters. The archived one-cluster result should now be read as a benchmark configuration that future sensitivity analysis must reproduce before changing one design axis at a time.

\section{Threshold cohorts and concentration}
The original threshold system partitions addresses by observed scale, attributed-record intensity, breadth, and attributed notional. The numerical cohorts are retained as archived arithmetic under the original convention.

\begin{table}[ht]
\centering
\caption{Archived threshold cohorts under record-level two-sided attribution. Whale status is an overlay.}
\begin{tabularx}{\textwidth}{>{\raggedright\arraybackslash}X r r r}
\toprule
Cohort & Addresses & Population share & Attributed-notional share \\
\midrule
Whale overlay & 68 & 0.09\% & 28.0\% \\
High-frequency operator & 2,952 & 3.82\% & 23.5\% \\
Power trader & 6,738 & 8.73\% & 29.9\% \\
Active lower-scale & 2,062 & 2.67\% & 10.6\% \\
High-breadth operator & 2,025 & 2.62\% & 1.1\% \\
Episodic low-notional & 63,358 & 82.07\% & 6.8\% \\
\midrule
High-scale/high-intensity subtotal & 9,758 & 12.64\% & 81.4\% \\
\bottomrule
\end{tabularx}
\end{table}

Define the archived concentration multiplier
\[
\chi_k^{\rm rec}=\frac{\text{attributed-notional share}_k}{\text{population share}_k}.
\]
For the high-scale/high-intensity subtotal, $\chi^{\rm rec}\approx6.44$; for the episodic low-notional cohort, $\chi^{\rm rec}\approx0.083$. The ratio of these average attribution intensities is approximately 77.

This is not a claim about the share of unique economic decisions, entity-level activity, or venue-wide notional. It is a concentration result under the archived record-attribution convention. Match normalization can change cohort membership because it directly changes intensity and average record-notional features.

\section{Microstructure signatures that remain usable}
\subsection{Breadth, timing, and concentration}
Market breadth, active-hour distribution, and market concentration remain deterministic properties of the archived address--record table, conditional on the contract and window filters. Their economic interpretation remains limited by address splitting and entity aggregation.

\subsection{Adverse-selection proxies}
For a record with validated execution direction, a signed markout can be written
\[
A_j=s_j(p_{j+\Delta}-p_j).
\]
Such a measure is valid only when $s_j$ has a well-defined economic direction. Mint and burn records require separate treatment rather than forced classification under a transfer convention.

\subsection{Price impact}
A reduced-form Kyle-type coefficient \citep{kyle1985} may be estimated from
\[
\Delta p_{m,t}=\alpha_m+\lambda_m q_{m,t}+\varepsilon_{m,t}.
\]
In thin bounded markets, raw estimates can be unstable. Winsorized values can be descriptive but should not be treated as structural parameters. Signed flow must also inherit the execution-type restrictions above.

\subsection{Surveillance candidates}
Large gross attributed volume combined with low net attributed position can generate candidate flags. It does not establish self-trading or manipulation because intermediation, hedging, cross-market arbitrage, fragmented matching, and multiple addresses can generate similar records.

\section{Required replication design}
A stronger replication should treat the archived configuration as rung zero and change one dimension at a time.

\subsection{Window axis}
Use:
\begin{enumerate}
\item the exact archived block range;
\item the full 21--27 April UTC week;
\item one month;
\item one quarter;
\item full contract history;
\item rolling and non-overlapping windows within each longer horizon.
\end{enumerate}
Rolling and non-overlapping windows are needed because a single nested history confounds participant persistence with venue and infrastructure changes.

\subsection{Attribution and direction axis}
At minimum compare:
\begin{enumerate}
\item the archived record-level two-sided attribution;
\item match-normalized counts and mean notional;
\item the feature vector with the direction proxy removed;
\item a mint/burn-aware direction specification;
\item a transfer-only direction specification.
\end{enumerate}

For each cell report cluster count, noise share, silhouette diagnostics, tier stability, transition matrices across windows, concentration, and address-level persistence. If the one-cluster null survives match normalization, execution-type separation, and longer windows, it becomes materially stronger. If clusters emerge, the original null remains a valid statement only about the archived representation.

\section{Implications for leveraged event markets}
The corrected evidence constrains, rather than directly calibrates, leveraged-event design.

First, concentrated address-level activity means a homogeneous synthetic trader population is a weak default. However, concentration parameters should be recalculated on match-normalized data before use in a risk engine.

Second, public books are not committed liquidation capacity, and fills do not reveal the cancellation process that determines whether depth survives stress.

Third, halt timing cannot be calibrated from address-level quote withdrawal using public fill data. At most, match-normalized fill-density decay can supply a weaker execution proxy.

Fourth, manipulation surveillance must be channel-specific. Fill anomalies can motivate investigation; quote-side allegations require quote-lifecycle evidence; oracle conduct requires resolution logs; outcome manipulation requires external-event evidence.

\section{Evidence discipline}
\begin{criterion}[Evidence-grade rule]
Every empirical statement should be assigned to one of four grades:
\begin{enumerate}
\item direct identity or arithmetic under a stated source and attribution convention;
\item deterministic feature or cohort derived under that convention;
\item descriptive statistical association;
\item role, intent, causal, or population-general interpretation requiring additional evidence.
\end{enumerate}
Only grades 1--3 are claimed as findings in this paper, and grade 2 must name the measurement convention on which it depends.
\end{criterion}

This rule and Proposition 2 are the durable contributions of the paper. They remain valid whether later replication preserves or rejects the archived clustering null.

\section{Limitations}
\textbf{Partial window.} The extraction covers approximately 2.285 days, not one week. Persistent archetypes cannot be inferred from this horizon.

\textbf{Legacy CTF only.} Negative-risk and multi-outcome market families are absent.

\textbf{Record attribution.} Both maker and taker are credited per \texttt{OrderFilled} record; intensity and mean record notional are not invariant to match fragmentation.

\textbf{Direction semantics.} The archived direction proxy does not fully distinguish conventional transfers, mints, and burns.

\textbf{Address pseudonymity.} Addresses are execution entities, not necessarily independent traders or firms.

\textbf{No address-level quote lifecycle.} Posted spread, quote lifetime, cancellations, quote stuffing, and spoofing are not identified.

\textbf{No causal identification.} Concentration and microstructure associations do not establish market-quality or welfare effects.

\section{Conclusion}
The corrected archival result is narrower than a population taxonomy of Polymarket participants. The study observes a dense 2.285-day slice of legacy-CTF execution records under a two-sided record-attribution convention. It documents strong attribution-weighted concentration and a one-cluster output under that representation. The feature space is affected by match fragmentation, execution-type semantics, contract-family exclusion, and the short horizon.

What survives these corrections is methodologically useful. Public fills do not reveal the quote lifecycle. A one-cluster output does not imply one economic type. Threshold cohorts are transparent measurement rules, not natural participant classes. Empirical claims must name both the observed data layer and the transformation that creates the feature space.

A match-normalized, mint-aware, multi-window rerun is therefore not a cosmetic robustness check; it is the required test of whether the archived null reflects the population or the measurement design.

\section*{Data and code statement}
This revision performs no new empirical run. It corrects the time interval, contract-family scope, unit-of-observation description, feature interpretation, and claim class of the archived results. The original numerical outputs are retained only as results under the legacy-CTF, record-level two-sided attribution convention. A future public release should provide exact block and contract registries, match-construction rules, execution-type classification, aggregate cohort outputs, processing manifests, and a reconstruction procedure without publishing ranked address-level profiles.

\printbibliography[heading=bibintoc,title={References}]
\end{document}